\documentclass[11pt]{article}
\usepackage{graphicx}
\usepackage{tabularx}
\usepackage{url}
\usepackage{amsthm}
\usepackage{amsmath}
\usepackage{amsfonts}
\usepackage[labelfont=bf]{caption}

\theoremstyle{plain}
\newtheorem{theorem}{Theorem}
\newtheorem{lemma}{Lemma}

\theoremstyle{definition}

\newtheorem{example}{Example}
\theoremstyle{remark}

\date{}

\begin{document}

\title{Lower Bounds on Cardinality of Reducts for Decision Tables from
Closed Classes}

\author{Azimkhon Ostonov and Mikhail Moshkov \\
Computer, Electrical and Mathematical Sciences \& Engineering Division \\ and Computational Bioscience Research Center\\
King Abdullah University of Science and Technology (KAUST) \\
Thuwal 23955-6900, Saudi Arabia\\ \{azimkhon.ostonov,mikhail.moshkov\}@kaust.edu.sa
}

\maketitle

\begin{abstract}
In this paper, we consider classes of decision tables closed under
removal of attributes (columns) and changing of decisions attached to rows.
For decision tables from closed classes, we study lower bounds on the minimum
cardinality of reducts, which are minimal sets of attributes that allow us
to recognize, for a given row, the decision attached to it. We assume that the
number of rows in decision tables from the closed class is not bounded from above by
a constant. We divide the set of such closed classes into two families. In
one family, only standard lower bounds $\Omega (\log $ ${\rm cl}(T))$ on the
minimum cardinality of reducts for decision tables hold, where ${\rm cl}(T)$
is the number of decision classes in the table $T$. In another family, these
bounds can be essentially tightened up to $\Omega ({\rm cl}(T)^{1/q})$ for
some natural $q$.
\end{abstract}


\section{Introduction\label{S1}}

Decision tables are a well-known way of presenting the information needed to
make decisions. These tables are used, in particular, in data analysis,
including classification problems, in modeling and studying problems related to combinatorial
optimization, fault diagnosis, computational geometry, etc. \cite%
{BorosHIK97,ChikalovLLMNSZ13,FurnkranzGL12,Humby73,Moshkov05,MoshkovZ11,Pawlak91,Pollack71,RokachM07}%
. Note that finite information systems with a selected decision attribute,
data sets with a selected class attribute, and partially defined Boolean
functions studied in various fields of data analysis as representations of
decision problems can naturally be interpreted as decision tables.

In this paper, we consider classes of decision tables closed under removal
of attributes (columns) and changing of decisions attached to rows. The most
natural examples of such classes are closed classes of decision tables
derived from information systems: the set of decision tables corresponding
to problems over an information system forms a closed class of decision
tables. However, the family of all closed classes of decision tables is
essentially wider than the family of closed classes derived from information
systems. In particular, the union of classes derived from two information systems is a
closed class, but generally, there is no an information system for which
this union is the closed class derived from it.

For decision tables from closed classes, we study lower bounds on the minimum
cardinality of (decision) reducts, which are minimal sets of attributes that
allow us to recognize, for a given row of the table, the decision attached to
it. Reducts are one of the main notions of rough set theory in which they
are used to choose appropriate features, to solve classification problems,
and to compress the knowledge  \cite{ChikalovLLMNSZ13,Pawlak91,PawlakS07,Slezak02,StawickiSJW17}. The bounds on the minimum
cardinality of reducts are of significant interest for rough set theory.

In this paper, we assume  that the number of rows in decision tables from the closed
class is not bounded from above by a constant. We divide the set of such
closed classes into two families. In one family, only standard lower bounds $%
\Omega (\log $ ${\rm cl}(T))$ on the minimum cardinality of reducts for
decision tables hold, where ${\rm cl}(T)$ is the number of decision classes
in the table $T$. In another family, these bounds can be essentially
tightened up to $\Omega ({\rm cl}(T)^{1/q})$ for some natural $q$. The
obtained results can be useful for the specialists in data analysis.

The present paper consists of six sections. Sections \ref{S2} and \ref{S3}
contain main definitions and some results related to the decision tables and to the closed
classes of decision tables. In Sect. \ref{S4}, we discuss lower bounds on
the cardinality of reducts and, in Sect. \ref{S5} -- examples related to the
closed classes of decision tables derived from information systems. Section %
\ref{S6} contains short conclusions.

\section{Decision Tables\label{S2}}

Let $B$ be a nonempty finite set with $k$ elements, $k\geq 2$. A $B$%
-decision table $T$ is a rectangular table with $n$ columns labeled with
attributes (really names of attributes) in which rows are pairwise different
tuples from $B^{n}$ that are labeled with nonnegative integers (decisions).
Denote by ${\rm Rows}(T)$ the set of rows of the table $T$, $N(T)$ the
number of rows in $T$, and ${\rm cl}(T)$ the number of different decisions
attached to rows of $T$ (the number of decision classes in the table $T$).
The number $n$ will be called the dimension of the table $T$ and will be
denoted $\dim T$.

A test for the table $T$ is a set of attributes (columns) of the table $T$
such that any two rows of the table $T$ labeled with different decisions are
different in at least one of the considered columns. A reduct for the table $%
T$ is a test for $T$ each proper subset of which is not a test. We denote by
$R(T)$ the minimum cardinality of a reduct for the table $T$. If ${\rm cl}(T)<2$, then $R(T)=0$.

Denote by $[T]$ the set of decision tables that can be obtained
from $T$ in the following way: we can remove from $T$ an arbitrary number of
attributes (columns), keep only one row from each group of equal rows in the
obtained table, and change in an arbitrary way decisions attached to the remaining
rows.

A decision table $T$ with $n$ columns will be called quasicomplete
if there exist two-elements subsets $B_{1},\ldots ,B_{n}$ of the set $B$
such that $$B_{1}\times \cdots \times B_{n}\subseteq {\rm Rows}(T).$$ We
denote by $I(T)$ the maximum dimension of a quasicomplete table from $[T]$.
The next statement follows immediately from Theorem 4.6 \cite{Moshkov05}.

\begin{lemma} \label{L1}
For any
$B$-decision table $T$ with ${\rm cl}(T)\ge 2$,
$$
N(T)\leq (k^{2}\dim T)^{I(T)}.
$$
\end{lemma}

\section{Closed Classes of Decision Tables\label{S3}}

Let $C$ be a set of $B$-decision tables. This set will be called a closed
class of decision tables if $C=\bigcup_{T\in C}[T]$. The closed class $C$
will be called nondegenerate if the number of rows in tables from $C$ is not
bounded from above by a constant.

We now define a parameter $I(C)$ of a nondegenerate closed class $C$ of decision
tables. If the parameter $I$ is bounded from above by a constant on tables
from the class $C$, then $I(C)=\max \{I(T):T\in C\}$. Otherwise, $%
I(C)=+\infty $.

Let us consider the behavior of the function $$N_{C}(n)=\max \{N(T):T\in C,\dim
T\leq n\}$$ defined on the set of natural numbers. This function characterizes the growth in the worst case of the number of rows in decision tables from the class $C$ with the growth of their dimension.

\begin{lemma} \label{L2}

Let $C$ be a nondegenerate closed class of $B$-decision tables.

{\rm (a)} If $I(C)<+\infty$, then $N_{C}(n)\leq
(k^{2}n)^{I(C)}$ for any natural $n$.

{\rm (b)} If $I(C)=+\infty $, then $2^{n}\leq
N_{C}(n)\leq k^{n}$ for any natural $n$.

\end{lemma}

\begin{proof}
(a) Let $I(C)<+\infty $. Using Lemma \ref{L1}, we obtain that $N_{C}(n)\leq
(k^{2}n)^{I(C)}$ for any natural $n$.

(b) Let $I(C)=+\infty $ and $n$ be a natural number.  The inequality $N_{C}(n)\leq k^{n}$ is obvious. Since $I(C)=+\infty $, there exists a quasicomplete table $%
T_{n} \in C$ with $\dim T_{n}=n$. It is clear that $N(T_{n})\geq 2^{n}$.
Therefore $2^{n}\leq N_{C}(n)$. 
\end{proof}

\section{ Bounds on Cardinality of Reducts \label{S4}}

First, we prove an auxiliary statement.

\begin{lemma} \label{L3}
Let $C$ be a nondegenerate closed class of $B$-decision tables and $T$ be a
decision table from $C$ with ${\rm cl}(T)\ge 2$. Then
$$
N_{C}(R(T))\geq {\rm cl}(T).
$$
\end{lemma}

\begin{proof}
Let $R(T)=m$ and $\{f_{1},\ldots ,f_{m}\}$ be a
reduct with the minimum cardinality for the table $T$. We
denote by $T^{\prime }$ a table from $[T]$, which is obtained from $T$ by
the removal of all attributes with the exception of $f_{1},\ldots ,f_{m}$.
Then the number of rows in the table $T^{\prime }$ should be at least the
number of decision classes in $T$, i.e., $N(T^{\prime })\geq {\rm cl}(T)$.
It is clear that $N(T^{\prime })\leq N_{C}(m)$. Therefore $N_{C}(m)\geq {\rm
cl}(T)$. 
\end{proof}

\begin{theorem} \label{T1}
Let $C$ be a nondegenerate closed class of $B$-decision tables.

{\rm (a)} If $I(C)<+\infty$, then $R(T)\geq {\rm cl}(T)^{1/I(C)}/k^{2}$ for any table  $T \in C$ with ${\rm cl}(T)\ge 2$.

{\rm (b)} If $I(C)=+\infty$, then $R(T)\geq \log _{k}{\rm cl}(T)$ for any table  $T \in C$ with ${\rm cl}(T)\ge 2$.

{\rm (c)} If $I(C)=+\infty$, then the inequality $R(T)\geq \log _{2}{\rm cl}(T)+1$ does not hold for infinitely many tables $T$ from the class $C$ for which both the dimension and the number of decision classes are not bounded from above by constants.
\end{theorem}

\begin{proof}
(a) Let $I(C)<+\infty $, $T \in C$, ${\rm cl}(T)\ge 2$, and $R(T)=m$. From Lemma \ref{L2} it follows that $N_{C}(m)\leq(k^{2}m)^{I(C)}$. By Lemma \ref{L3},  $N_{C}(m)\geq {\rm cl}(T)$. Therefore
$(k^{2}m)^{I(C)}\geq {\rm cl}(T)$ and  $m\geq {\rm cl}
(T)^{1/I(C)}/k^{2}$.

(b) Let $I(C)=+\infty $, $T \in C$, ${\rm cl}(T)\ge 2$, and $R(T)=m$. From Lemma \ref{L2} it follows that $N_{C}(m)\leq k^{m}$. By Lemma \ref{L3},  $N_{C}(m)\geq {\rm cl}(T)$.
Therefore $k^{m}\geq {\rm cl}(T)$ and $m\geq \log _{k}{\rm cl}(T)$.

(c)  Let $n$ be a natural number. Since $I(C)=+\infty $,
there exists a quasicomplete decision table $T_{n}$ from $C$ with $\dim
T_{n}=n$ and ${\rm cl}(T_{n})\geq 2^{n}$. Let us assume that $R(T_{n})\geq
\log _{2}{\rm cl}(T)+1$. Then $R(T_n)\geq \log _{2}2^{n}+1=n+1$. It is
obvious, that $n\geq R(T_{n})$. Thus, the inequality $R(T_{n})\geq \log _{2}%
{\rm cl}(T_{n})+1$ does not hold.
\end{proof}

The statement (c) shows that the bound from the statement (b) cannot be improved essentially.

\section{Closed Classes of Decision Tables Derived from Information Systems \label{S5}}

The most natural examples of closed classes of decision tables are classes
derived from infinite information systems. An infinite information
system is a triple $U=(A,F,B)$, where $A$ is an infinite set of objects
called universe, $B$ is a finite set with $k$ elements, $k\geq 2$, and $F$
be an infinite set of functions from $A$ to $B$ called attributes. A problem
over $U$ is specified by a finite number of attributes $f_{1},\ldots
,f_{n}\in F$ that divide the universe $A$ into nonempty domains in each of
which values of attributes $f_{1},\ldots ,f_{n}$ are fixed. Each domain is
labeled with a decision. For a given object $a\in A$, it is required to
recognize the decision attached to the domain to which the object $a$ belongs.
A decision table corresponds to this problem in the following way: the table
contains $n$ columns labeled with attributes $f_{1},\ldots ,f_{n}$, rows of
this table correspond to domains and are labeled with decisions attached to
the domains.

We denote by ${\rm Tab}(U)$ the set of decision tables corresponding to all
problems over the information system $U$. One can show that ${\rm Tab}(U)$
is a nondegenerate closed class of decision tables. We will say that this
class is derived from the information system $U$.

A subset $\{f_{1},\ldots ,f_{p}\}$ of the set $F$ is called independent if
there exist two-element subsets $B_{1},\ldots ,B_{p}$ of the set $B$ such
that, for any tuple $(b_{1},\ldots ,b_{p})\in B_{1}\times \cdots \times
B_{p} $, the equations system $$\{f_{1}(x)=b_{1},\ldots ,f_{p}(x)=b_{p}\}$$
has a solution from $A$. If, for any natural $p$, the set $F$
contains an independent subset, which cardinality is equal to $p$, then $I(%
{\rm Tab}(U))=+\infty $. Otherwise, $I({\rm Tab}(U))$ is the maximum
cardinality of an independent subset of the set $F$.

We now consider examples of infinite information systems from the book
\cite{MoshkovZ11}.

\begin{example}
Let $P$ be the Euclidean plane and $l$ be a straight line in $P$. This line
divides the plane into two open half-planes $h_{1}$ and $h_{2}$, and the line
$l$. We correspond an attribute to the line $l$. This attribute takes the
value $0$ on points from $h_{1}$ and the value $1$ on points from $h_{2}$
and $l$. We denote by $F_{P}$ the set of attributes corresponding to all
lines in $P$ and consider the information system $U_{p}=(P,F_{P},\{0,1\})$.
There are two lines that divide the plane $P$ into four domains, but there
are no three lines that divide $P$ into eight domains. Therefore $I({\rm Tab}%
(U_{P}))=2$ and, for any table $T\in {\rm Tab}(U_{P})$ with ${\rm cl}(T)\ge 2$, $R(T)\geq
{\rm cl}(T)^{1/2}/4$. This lower bound is essentially tighter than the
standard bound $R(T)\geq \log _{2}{\rm cl}(T)$.
\end{example}

\begin{example}
Let $m$ and $t$ be natural numbers. We denote by ${\rm Pol}(m)$ the set of
polynomials with integer coefficients that depend on variables $x_{1},\ldots
,x_{m}$. We denote by ${\rm Pol}(m,t)$ the set of polynomials from $Pol(m)$,
which degree is at most $t$. We define information systems $U(m)$ and $%
U(m,t) $ in the following way: $U(m)=(%
\mathbb{R}
^{m},F(m),E)$ and $U(m,t)=(%
\mathbb{R}
^{m},F(m,t),E),$ where $%
\mathbb{R}
$ is the set of real numbers, $E=\{-1,0,+1\}$, $F(m)=\{{\rm sign}(p):p\in {\rm Pol}(m)\}$, $%
F(m,t)=\{{\rm sign}(p):p\in {\rm Pol}(m,t)\}$, and ${\rm sign}(x)=-1$ if $%
x<0 $, ${\rm sign}(x)=0$ if $x=0$, and ${\rm sign}(x)=+1$ if $x>0$. One can
show that $I({\rm Tab}(U(m)))= +\infty $ and $I({\rm Tab}(U(m,t)))<+\infty $.
Therefore, for any natural $m$ and any table $T$ from ${\rm Tab}(U(m))$ with ${\rm cl}(T)\ge 2$, $%
R(T)\geq \log _{3}cl(T)$ and this bound cannot be tightened essentially. For
any natural $m$ and $t$ and any table $T$ from ${\rm Tab}(U(m,t))$ with ${\rm cl}(T)\ge 2$, $R(T)\geq
{\rm cl}(T)^{1/q}/9$ for some natural $q$.
\end{example}

\section{Conclusions\label{S6}}

In this paper, we divided the set of nondegenerate closed classes of
decision tables into two families. For closed classes from one family, only standard lower bounds
$\Omega (\log $ ${\rm cl}(T))$ on the minimum cardinality of reducts for
decision tables hold, where ${\rm cl}(T)$ is the number of decision classes
in the table $T$. For closed classes from  another family, these bounds can be essentially
tightened up to $\Omega ({\rm cl}(T)^{1/q})$ for some natural $q$.

\subsection*{Acknowledgements}

Research reported in this publication was supported by King Abdullah
University of Science and Technology (KAUST).

\bibliographystyle{spmpsci}
\bibliography{abc_bibliography}

\end{document}